\documentclass[12pt]{article}

\makeatletter
\let\@fnsymbol\@arabic
\makeatother

\usepackage{amsmath,amsthm,amsfonts,amscd}
\usepackage{fullpage}
\usepackage{float}
\usepackage{url}
\usepackage{graphicx}

\DeclareMathOperator{\Th}{Th}

\DeclareMathOperator{\Pal}{Pal}
\DeclareMathOperator{\tabb}{TAB}
\def\Enn{{\mathbb{N}}}
\def\Zee{{\mathbb{Z}}}

\usepackage{pgf}
\usepackage{tikz}
\usetikzlibrary{arrows,automata,positioning}

\def \nodiv{{\, |\kern-4.5pt/}\, }

\def\modd#1 #2{#1\ ({\rm mod}\ #2)}

\newcommand{\st}{\;:\;}

\makeatletter
\def\Ddots{\mathinner{\mkern1mu\raise\p@
\vbox{\kern7\p@\hbox{.}}\mkern2mu
\raise4\p@\hbox{.}\mkern2mu\raise7\p@\hbox{.}\mkern1mu}}
\makeatother

\author{Hamoon Mousavi \& Jeffrey Shallit\thanks{School of Computer Science,
University of Waterloo,
Waterloo,  ON  N2L 3G1,
Canada; \newline
{\tt sh2mousa@uwaterloo.ca},
        {\tt shallit@uwaterloo.ca} .}}

\title{Mechanical Proofs of Properties of the Tribonacci Word}

\begin{document}

\maketitle

\theoremstyle{plain}
\newtheorem{theorem}{Theorem}
\newtheorem{corollary}[theorem]{Corollary}
\newtheorem{almosttheorem}[theorem]{(Almost) Theorem}
\newtheorem{lemma}[theorem]{Lemma}
\newtheorem{proposition}[theorem]{Proposition}

\theoremstyle{definition}
\newtheorem{definition}[theorem]{Definition}
\newtheorem{example}[theorem]{Example}
\newtheorem{conjecture}[theorem]{Conjecture}
\newtheorem{openproblem}[theorem]{Open Problem}
\newtheorem{proc}[theorem]{Procedure}

\theoremstyle{remark}
\newtheorem{remark}[theorem]{Remark}

\begin{abstract}
We implement a decision procedure for answering questions about
a class of infinite words that might be called (for lack of a better
name) ``Tribonacci-automatic''.  This class includes, for example,
the famous Tribonacci word ${\bf T} = 0102010010201\cdots$,
the fixed point of the 
morphism $0 \rightarrow 01$,
$1 \rightarrow 02$, $2 \rightarrow 0$.    We use it to reprove some
old results
about the Tribonacci
word from the literature,
such as assertions about the
occurrences in $\bf T$ of squares, cubes, palindromes, and so forth.
We also obtain some new results.
\end{abstract}

Note:  some sections of this paper have been taken, more or less
verbatim, from another preprint of the authors and C. F. Du and
L. Schaeffer \cite{Du&Mousavi&Schaeffer&Shallit:2014}.

\section{Decidability}
\label{decide}

As is well-known, the logical theory $\Th(\Enn,+)$, sometimes called
Presburger arithmetic, is decidable \cite{Presburger:1929,Presburger:1991}.
B\"uchi \cite{Buchi:1960} showed that if we add
the function $V_k(n) = k^e$, for some fixed integer $k \geq 2$,
where $e = \max \{ i \ : \ k^i \, | \, n \}$,
then the resulting theory is still decidable. 
This theory is powerful enough to define finite automata;
for a survey, see \cite{Bruyere&Hansel&Michaux&Villemaire:1994}.  

As a consequence, we have the following theorem 
(see, e.g., \cite{Shallit:2013}):
\begin{theorem}
There is an algorithm that, given a proposition phrased using only the
universal and existential quantifiers, indexing into one or more $k$-automatic
sequences, addition, subtraction, logical operations, and
comparisons, will decide the truth of that proposition.
\label{one}
\end{theorem}
Here, by a $k$-automatic sequence, we mean a sequence $\bf a$
computed by deterministic finite automaton with
output (DFAO) $M = (Q, \Sigma_k, \Delta, \delta, q_0, \kappa) $.
Here $\Sigma_k := \lbrace 0,1,\ldots, k-1 \rbrace$ is the input 
alphabet,
$\Delta$ is the output alphabet,
and outputs are associated with the states given by the map
$\kappa:Q \rightarrow \Delta$ in the following manner:  if $(n)_k$ denotes
the canonical expansion of $n$ in base $k$, then
${\bf a}[n] = \kappa(\delta(q_0, (n)_k))$.  The prototypical example of
an automatic sequence is the Thue-Morse sequence
${\bf t} = t_0 t_1 t_2 \cdots$, the fixed point (starting with $0$) of
the morphism $0 \rightarrow 01$, $1 \rightarrow 10$.

It turns out that many results in the literature about properties of automatic
sequences, for which some had only
long and involved proofs, can be proved purely mechanically using a decision
procedure.
It suffices to express the property as an appropriate logical
predicate, convert the predicate into an automaton accepting
representations of integers for which the predicate is true, and
examine the automaton.
See, for example, the recent papers
\cite{Allouche&Rampersad&Shallit:2009,Goc&Henshall&Shallit:2012,Goc&Saari&Shallit:2013,Goc&Mousavi&Shallit:2013,Goc&Schaeffer&Shallit:2013}. 
Furthermore, in many cases we can explicitly enumerate various aspects
of such sequences, such as subword complexity
\cite{Charlier&Rampersad&Shallit:2012}.

Beyond base $k$, more exotic numeration systems are known, and one
can define automata taking representations in these systems as input.
It turns out that in the so-called Pisot numeration systems, addition
is computable \cite{Frougny:1992a,Frougny&Solomyak:1996},
and hence a theorem analogous to Theorem~\ref{one} holds
for these systems.  See, for example, \cite{Bruyere&Hansel:1997}.  
It is our contention that the power of this approach has not been
widely appreciated, and that
many results, previously proved using long and involved ad hoc techniques,
can be proved with much less effort by phrasing them as logical predicates
and employing a decision procedure.  Furthermore, many enumeration questions
can be solved with a similar approach.

In a previous paper, we explored the consequences of a decision algorithm
for Fibonacci representation \cite{Du&Mousavi&Schaeffer&Shallit:2014}.
In this paper we discuss 
our implementation of an analogous algorithm 
for Tribonacci representation.   
We use it to
reprove some old results from the literature purely mechanically, as well as
obtain some new results.

For other works on using computerized formal methods to prove
theorems see, for example, \cite{Hales:2008,Konev&Lisitsa:2014}.

\section{Tribonacci representation}
\label{fibrep}

Let the Tribonacci numbers be defined, as usual, by the linear
recurrence
$T_n = T_{n-1} + T_{n-2}+ T_{n-3}$ for $n \geq 3$
with initial values
$T_0 = 0$, $T_1 = 1$, $T_2 = 1$.
(We caution the reader that some authors use a
different indexing for these numbers.)  
Here are the first few values of this sequence.
\begin{table}[H]
\begin{center}
\begin{tabular}{c|ccccccccccccccccc}
$n$ & 0 & 1 & 2 & 3 & 4 & 5 & 6 & 7 & 8 & 9 & 10 & 11 & 12 & 13 & 14 & 15 & 16 \\
\hline
$T_n$ & 0& 1& 1& 2& 4& 7& 13& 24& 44& 81& 149& 274& 504& 927& 1705& 3136& 5768 \\
\end{tabular}
\end{center}
\end{table}
From the theory of linear recurrences we know that 
$$T_n = c_1 \alpha^n + c_2 \beta^n + c_3 \gamma^n$$
where $\alpha, \beta, \gamma$ are the zeros of the
polynomial $x^3 - x^2 -x - 1$.  The only real zero is
$\alpha \doteq 1.83928675521416113255185$; the other two zeros
are complex and are of magnitude $<3/4$.
Solving for the constants, we find that
$c_1 \doteq 0.336228116994941094225362954$, the real 
zero of the polynomial $44x^3-2x-1=0$.
It follows that $T_n = c_1 \alpha^n + O(.75^n)$.
In particular $T_n/T_{n-1} = \alpha + O(.41^n)$.

It is well-known that every 
non-negative integer can be represented, in an essentially unique
way, as a sum of Tribonacci numbers $(T_i)_{i\geq 2}$,
subject to the constraint that no three consecutive Tribonacci numbers
are used 
\cite{Carlitz&Scoville&Hoggatt:1972}.
For example, $43 = T_7 + T_6 + T_4 + T_3$.

Such a representation can be written as a binary word
$a_1 a_2 \cdots a_n$ representing
the integer
$\sum_{1 \leq i \leq n} a_i T_{n+2-i}$.  For example,
the binary word $110110$ is the Tribonacci representation of $43$.

For $w = a_1 a_2 \cdots a_n \in \Sigma_2^*$, we
define $[a_1 a_2 \cdots a_n]_T := \sum_{1 \leq i \leq n} a_i T_{n+2-i}$,
even if $a_1 a_2 \cdots a_n$ has leading zeros or occurrences of the
word $111$.

By $(n)_T$ we mean the {\it canonical} Tribonacci representation for
the integer $n$, having no leading zeros or occurrences of
$111$.  Note
that $(0)_T = \epsilon$, the empty word.  The language of all
canonical representations of elements of $\Enn$ is 
$\epsilon + (1 + 11)(0+01+011)^*$.

Just as Tribonacci representation is an analogue of base-$k$ representation,
we can define the notion of {\it Tribonacci-automatic sequence} as the
analogue of the more familiar notation of $k$-automatic sequence
\cite{Cobham:1972,Allouche&Shallit:2003}.  We say that an infinite word
${\bf a} = (a_n)_{n \geq 0}$ is Tribonacci-automatic if there exists an
automaton with output $M = (Q, \Sigma_2, q_0, \delta, \kappa, \Delta)$
that $a_n = \kappa(\delta(q_0, (n)_T))$ for all $n \geq 0$.  An example
of a Tribonacci-automatic sequence is the infinite Tribonacci word,
$${\bf T} = T_0 T_1 T_2 \cdots =  0102010010201\cdots$$
which is generated by the following 3-state automaton:

\begin{figure}[H]
\begin{center}
\includegraphics[width=4.5in]{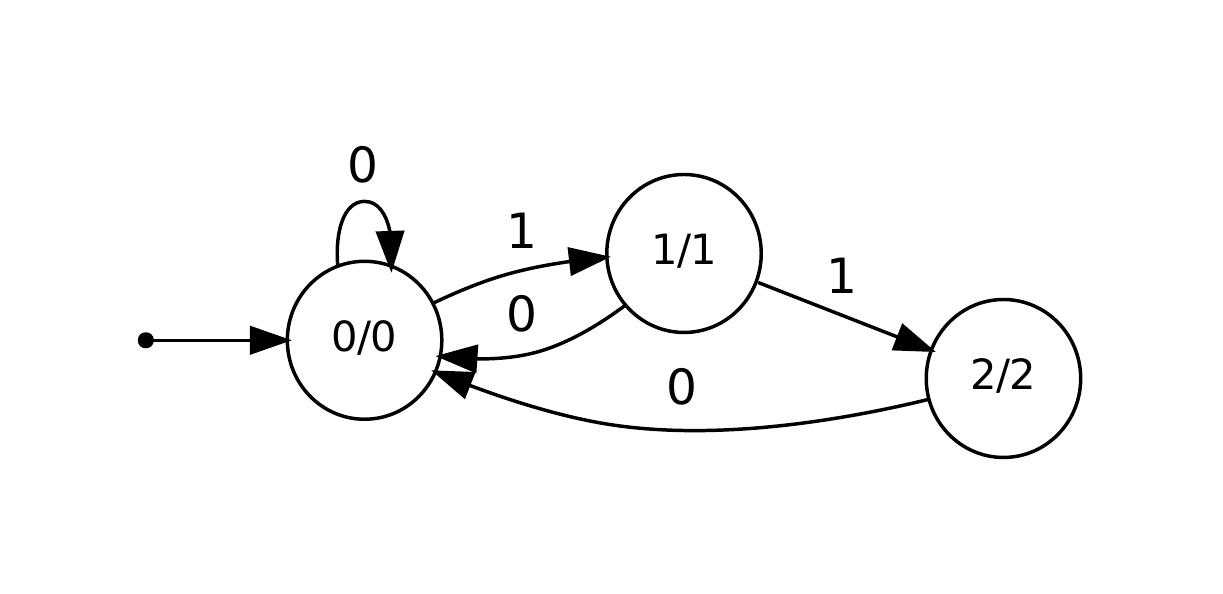}
\caption{Automaton generating the Tribonacci sequence}
\label{tribonaccis}
\end{center}
\end{figure}

To compute $T_i$, we express $i$ in canonical Tribonacci representation,
and feed it into the automaton.  Then $T_i$ is the output associated with
the last state reached (denoted by the symbol after the slash).

A basic fact about Tribonacci representation is that addition can
be performed by a finite automaton.  To make this precise, we need to
generalize our notion of Tribonacci representation to $r$-tuples of
integers for $r \geq 1$.  A representation for $(x_1, x_2,\ldots, x_r)$
consists of a string
of symbols $z$ over the alphabet $\Sigma_2^r$, such that the projection
$\pi_i(z)$ over the $i$'th coordinate gives a Tribonacci representation
of $x_i$.  Notice that since the canonical
Tribonacci representations of the individual $x_i$
may have different lengths, padding with leading zeros will often
be necessary.  A representation for $(x_1, x_2, \ldots, x_r)$ is called
canonical if it has no leading $[0,0,\ldots 0]$ symbols and the projections
into individual coordinates have no occurrences of $111$.  We write the
canonical representation 
as $(x_1, x_2, \ldots, x_r)_T$.  Thus,
for example, the canonical representation for $(9,16)$ is
$[0,1][1,0][0,0][1,1][0,1]$.

Thus, our claim about addition in Tribonacci representation is that there
exists a 
deterministic finite automaton (DFA) $M_{\rm add}$
that takes input words of the form $[0,0,0]^* (x,y,z)_T$,
and accepts if and only if $x +y =z$.
Thus, for example, $M_{\rm add}$ accepts $[1,0,1][0,1,1][0,0,0]$
since the three words obtained by projection are $100$, $010$, and $110$,
which represent, respectively, $4$, $2$, and $6$ in Tribonacci representation.

Since this automaton does not appear to have been given explicitly in
the literature and it is essential to our implementation,
we give it here.  This automaton
actually works even for non-canonical expansions having
three consecutive $1$'s.
The initial state is state $1$.
The state $0$ is a ``dead
state'' that can safely be ignored.

We briefly sketch a proof of the correctness of this automaton.  
States can be identified with certain sequences, as follows:
if $x,y,z$ are the identical-length words arising from projection
of a word that takes $M_{\rm add}$ from the initial state $1$ to the
state $t$, then $t$ is identified with the integer sequence
$([x0^n]_T + [y0^n]_T - [z0^n]_T)_{n \geq 0}$.    State $0$ corresponds
to sequences that can never lead to $0$, as they are too positive or
too negative.

When we intersect this automaton with the appropriate regular language
(ruling out input triples containing $111$ in any coordinate), we get
an automaton with 149 states accepting $0^*(x,y,z)_T$ such that
$x + y = z$.

Another basic fact about Tribonacci representation is that, for
canonical representations containing no three consecutive $1$'s or
leading zeros, the radix order on representations is the same 
as the ordinary ordering on $\Enn$.    It follows that a very
simple automaton can, on input $(x,y)_T$, decide whether $x < y$.

Putting this all together, we get the analogue of
Theorem~\ref{one}:

\begin{proc}[Decision procedure for Tribonacci-automatic words] \label{proc:Fib-auto-decide} \ \\
{\bf Input:} \begin{itemize}
	\item $m,n \in \Enn$;
	\item $m$ DFAOs generating the Tribonacci-automatic words ${\bf w}_1,{\bf w}_2,\dots,{\bf w}_m$;
	\item a first-order proposition with $n$ free variables $\varphi(v_1,v_2,\dots,v_n)$ using constants and relations definable in $\Th(\Enn,0,1,+)$ and indexing into ${\bf w}_1,{\bf w}_2,\dots,{\bf w}_m$. 
	\end{itemize}
{\bf Output:} DFA with input alphabet $\Sigma_2^n$ accepting $\{ (k_1,k_2,\dots,k_n)_T \st \varphi(k_1,k_2,\dots,k_n) \text{ holds} \}$.
\end{proc}

\begin{table}[H]
\begin{center}
\begin{tabular}{c|cccccccc|c}
$q$ & [0,0,0] & [0,0,1] & [0,1,0] & [0,1,1] & [1,0,0] & [1,0,1] & [1,1,0] & [1,1,1] & acc/rej\\
\hline
0& 0& 0& 0& 0& 0& 0& 0& 0& 0\\
1& 1& 2& 3& 1& 3& 1& 0& 3& 1\\
2& 4& 0& 5& 4& 5& 4& 6& 5& 0\\
3& 0& 7& 0& 0& 0& 0& 0& 0& 0\\
4& 0& 0& 0& 0& 0& 0& 8& 0& 0\\
5& 9& 0&10& 9&10& 9&11&10& 0\\
6&12&13& 0&12& 0&12& 0& 0& 1\\
7& 0&14& 0& 0& 0& 0& 0& 0& 0\\
8& 0& 0& 9& 0& 9& 0&10& 9& 0\\
9& 0& 0& 4& 0& 4& 0& 5& 4& 0\\
10& 2&15& 1& 2& 1& 2& 3& 1& 0\\
11& 7&16& 0& 7& 0& 7& 0& 0& 1\\
12&14&17& 0&14& 0&14& 0& 0& 1\\
13&18&19&20&18&20&18&21&20& 0\\
14& 3& 1& 0& 3& 0& 3& 0& 0& 0\\
15& 0& 0& 0& 0& 0& 0&22& 0& 0\\
16&20&18&21&20&21&20& 0&21& 1\\
17& 5& 4& 6& 5& 6& 5&23& 6& 1\\
18& 0& 0& 8& 0& 8& 0&24& 8& 0\\
19& 0& 0& 0& 0& 0& 0&25& 0& 0\\
20&10& 9&11&10&11&10& 0&11& 1\\
21& 0&12& 0& 0& 0& 0& 0& 0& 0\\
22& 0& 0&26& 0&26& 0&27&26& 0\\
23& 0&28& 0& 0& 0& 0& 0& 0& 0\\
24&13&29&12&13&12&13& 0&12& 0\\
25& 0& 0& 0& 0& 0& 0&26& 0& 0\\
26& 0& 0& 0& 0& 0& 0& 4& 0& 0\\
27&15& 0& 2&15& 2&15& 1& 2& 0\\
28& 0&30& 0& 0& 0& 0& 0& 0& 0\\
29& 0& 0&31& 0&31& 0&32&31& 0\\
30& 0& 3& 0& 0& 0& 0& 0& 0& 0\\
31& 0& 0& 0& 0& 0& 0&33& 0& 0\\
32&26& 0&27&26&27&26&34&27& 0\\
33& 0& 0& 0& 0& 0& 0& 9& 0& 0\\
34&16&35& 7&16& 7&16& 0& 7& 0\\
35&31& 0&32&31&32&31&36&32& 0\\
36&37&38&39&37&39&37& 0&39& 1\\
37&17&40&14&17&14&17& 0&14& 0\\
38&19& 0&18&19&18&19&20&18& 0\\
39& 0&41& 0& 0& 0& 0& 0& 0& 1\\
40& 0& 0&22& 0&22& 0&42&22& 0\\
41&21&20& 0&21& 0&21& 0& 0& 0\\
42&38&43&37&38&37&38&39&37& 0\\
43& 0& 0& 0& 0& 0& 0&31& 0& 0\\
\hline
\end{tabular}
\end{center}
\caption{Transition table for $M_{\rm add}$ for Tribonacci addition}
\end{table}

\section{Mechanical proofs of properties of the infinite Tribonacci word}
\label{proofsf}

Recall that a word $x$, whether finite or infinite, is said to have
period $p$ if $x[i] = x[i+p]$ for all $i$ for which this equality is meaningful.
Thus, for example, the English word ${\tt alfalfa}$ has period $3$.
The {\it exponent} of a finite word $x$, written $\exp(x)$, is
$|x|/P$, where $P$ is the smallest period of $x$.
Thus $\exp({\tt alfalfa}) = 7/3$.

If $\bf x$ is an infinite word with a finite period, we say it is
{\it ultimately periodic}.  An infinite
word $\bf x$ is ultimately periodic if and only if there are finite
words $u, v$ such that $x = uv^\omega$, where $v^\omega= vvv \cdots$.

A nonempty word of the form $xx$ is called a {\it square}, and a 
nonempty word of
the form $xxx$ is called a {\it cube}.  More generally, a nonempty word
of the form $x^n$ is called an $n$'th power.
By the {\it order} of a square $xx$,
cube $xxx$, or $n$'th power $x^n$, we mean the length $|x|$.

The infinite Tribonacci word ${\bf T} = 0102010 \cdots = T_0 T_1 T_2 \cdots$
can be described in many
different ways.  In addition to our definition in terms of automata,
it is also the fixed point of 
the morphism $\varphi(0) = 01$, $\varphi(1) = 02$, and $\varphi(1) = 0$.
This word
has been studied extensively in the literature; see, for example,
\cite{Chekhova&Hubert&Messaoudi:2001,Barcucci&Belanger&Brlek:2004,Rosema&Tijdeman:2005,Glen:2006,Tan&Wen:2007,Duchene&Rigo:2008,Richomme&Saari&Zamboni:2010,Turek:2013}.

It can also be described as the limit of the finite Tribonacci
words $(Y_n)_{n \geq 0}$, defined as follows:
\begin{align*}
Y_0 &= \epsilon \\
Y_1 &= 2 \\
Y_2 &= 0 \\
Y_3 &= 01 \\
Y_n &= Y_{n-1} Y_{n-2} Y_{n-3} \ \text{for $n \geq 4$}.
\end{align*}
Note that $Y_n$, for $n \geq 2$, is the prefix of length $T_n$ of
$\bf T$.

In the next subsection, we use our implementation to prove a variety of
results about repetitions in $\bf T$.

\subsection{Repetitions}
\label{repe-subsec}

It is known that all strict epistandard words (or Arnoux-Rauzy words),
are not ultimately periodic (see, for example, \cite{Glen&Justin:2009}).
Since $\bf T$ is in this class, we have the following known result which
we can reprove using our method.

\begin{theorem}
The word $\bf T$ is not ultimately periodic.
\end{theorem}

\begin{proof}
We construct a predicate asserting that the integer
$p \geq 1$ is a period of some suffix of $\bf T$:
$$ (p \geq 1) \ \wedge \ \exists n \ \forall i \geq n\  {\bf T}[i] =
{\bf T}[i+p] . $$
(Note:  unless otherwise indicated, whenever we refer to a variable in
a predicate, the range of the variable is assumed to be
$\Enn = \lbrace 0, 1, 2, \ldots \rbrace$.)
From this predicate, using our program, we constructed
an automaton accepting the language
$$ L = 0^*\ \lbrace (p)_T \ : \ (p \geq 1) \ \wedge \ \exists n
\ \forall i \geq n \ {\bf T}[i] = {\bf T}[i+p] \rbrace .$$
This automaton accepts the empty language, and so it follows
that ${\bf T}$ is not ultimately periodic.


Here is the log of our program:
\begin{verbatim}
p >= 1 with 5 states, in 426ms
 i >= n with 13 states, in 3ms
  i + p with 150 states, in 31ms
   TR[i] = TR[i + p] with 102 states, in 225ms
    i >= n => TR[i] = TR[i + p] with 518 states, in 121ms
     Ai i >= n => TR[i] = TR[i + p] with 4 states, in 1098ms
      En Ai i >= n => TR[i] = TR[i + p] with 2 states, in 0ms
       p >= 1 & En Ai i >= n => TR[i] = TR[i + p] with 2 states, in 1ms
overall time: 1905ms
\end{verbatim}
The largest intermediate automaton during the computation had 5999 states.

A few words of explanation are in order:  here ``{\tt T}'' 
refers to the sequence
$\bf T$, and ``{\tt E}'' is our abbreviation for $\exists$ and
``{\tt A}'' is our abbreviation for $\forall$.  The symbol ``{\tt =>}'' 
is logical
implication, and ``{\tt \&}'' is logical and.
\end{proof}

From now on, whenever we discuss the language accepted by an automaton,
we will omit the $0^*$ at the beginning.  

We now turn to repetitions.  As a particular case of
\cite[Thm.~6.31 and Example 7.6, p.~130]{Glen:2006} and \cite[Example
6.21]{Glen:2007} we have the following result, which we can reprove
using our method.

\begin{theorem}
$\bf T$ contains no fourth powers.
\end{theorem}

\begin{proof}
A natural predicate for the orders of all fourth powers occurring
in $\bf T$:
$$(n > 0) \ \wedge \ \exists i \ \forall t<3n \  {\bf T}[i+t] = {\bf T}[i+n+t] .
$$
However, this predicate could not be run on our prover.  It runs out
of space while trying to determinize an NFA with 24904 states.

Instead, we make the substitution $j = i+t$, obtaining the new predicate
$$ (n > 0) \ \wedge \ \exists i \ \forall j\ 
	((j \geq i) \wedge (j<i+3n)) \implies  {\bf T}[j] = {\bf T}[j+n] .
$$

The resulting automaton accepts nothing, so there are no fourth powers.

Here is the log.
{\footnotesize
\begin{verbatim}
n > 0 with 5 states, in 59ms
 i <= j with 13 states, in 15ms
  3 * n with 147 states, in 423ms
   i + 3 * n with 799 states, in 4397ms
    j < i + 3 * n with 1103 states, in 4003ms
     i <= j & j < i + 3 * n with 1115 states, in 111ms
      j + n with 150 states, in 18ms
       TR[j] = TR[j + n] with 102 states, in 76ms
        i <= j & j < i + 3 * n => TR[j] = TR[j + n] with 6550 states, in 1742ms
         Aj i <= j & j < i + 3 * n => TR[j] = TR[j + n] with 4 states, in 69057ms
          Ei Aj i <= j & j < i + 3 * n => TR[j] = TR[j + n] with 2 states, in 0ms
           n > 0 & Ei Aj i <= j & j < i + 3 * n => TR[j] = TR[j + n] with 2 states, in 0ms
overall time: 79901ms
\end{verbatim}
}
The largest intermediate automaton in the computation had 86711 states.
\end{proof}

Next, we move on to a description of the orders of squares occurring
in $\bf T$.    We reprove a result of Glen \cite[\S 6.3.5]{Glen:2006}.

\begin{theorem}
All squares in $\bf T$ are of order $T_n$ or $T_n + T_{n-1}$
for some $n \geq 2$.
Furthermore, for all $n \geq 2$, there exists a square of order
$T_n$ and $T_n + T_{n-1}$ in $\bf T$.
\label{squares}
\end{theorem}

\begin{proof}
A natural predicate for the lengths of squares is
$$(n > 0) \ \wedge \ \exists i \ \forall t<n \  {\bf T}[i+t] = {\bf T}[i+n+t] .$$
but when we run our solver on this predicate, we get an intermediate
NFA of 4612 states that our solver could not determinize in the the
allotted space.  The problem appears to arise from the three 
different variables indexing $T$.  To get around this problem, we
rephrase the predicate, introducing a new variable $j$ that represents
$i+t$.  This gives the predicate
$$ (n>0) \ \wedge \ \exists i \ \forall j \ ((i\leq j) \wedge (j< i+n))
	\implies {\bf T}[j] = {\bf T}[j+n] .$$
and the following log
\begin{verbatim}
i <= j with 13 states, in 10ms
 i + n with 150 states, in 88ms
  j < i + n with 229 states, in 652ms
   i <= j & j < i + n with 241 states, in 42ms
    j + n with 150 states, in 19ms
     TR[j] = TR[j + n] with 102 states, in 61ms
      i <= j & j < i + n => TR[j] = TR[j + n] with 1751 states, in 341ms
       Aj i <= j & j < i + n => TR[j] = TR[j + n] with 11 states, in 4963ms
        Ei Aj i <= j & j < i + n => TR[j] = TR[j + n] with 4 states, in 4ms
         n > 0 & Ei Aj i <= j & j < i + n => TR[j] = TR[j + n] with 4 states, in 0ms
overall time: 6232ms
\end{verbatim}
The resulting automaton accepts
exactly the language $10^* + 110^*$.    
The largest intermediate automaton had 26949 states.
\end{proof}

We can easily get more information about the square
occurrences in $\bf T$.    By modifying our previous predicate, we get
$$ (n>0) \ \wedge \ \forall j \ ((i\leq j) \wedge (j< i+n))
        \implies {\bf T}[j] = {\bf T}[j+n]$$
which encodes those $(i,n)$ pairs such that there is a square of order
$n$ beginning at position $i$ of $\bf T$.

This automaton has only 10 states and efficiently encodes the
orders and starting positions of each square in $\bf T$.  During
the computation, the largest intermediate automaton had 
26949 states.
Thus we have proved

\begin{theorem}
The language
$$ \{ (i,n)_T \ : \ \text{there is a square of order $n$ beginning at
	position $i$ in {\bf T}} \}$$
is accepted by the automaton in Figure~\ref{squareorders}.
\begin{figure}[H]
\begin{center}
\includegraphics[width=6.5in]{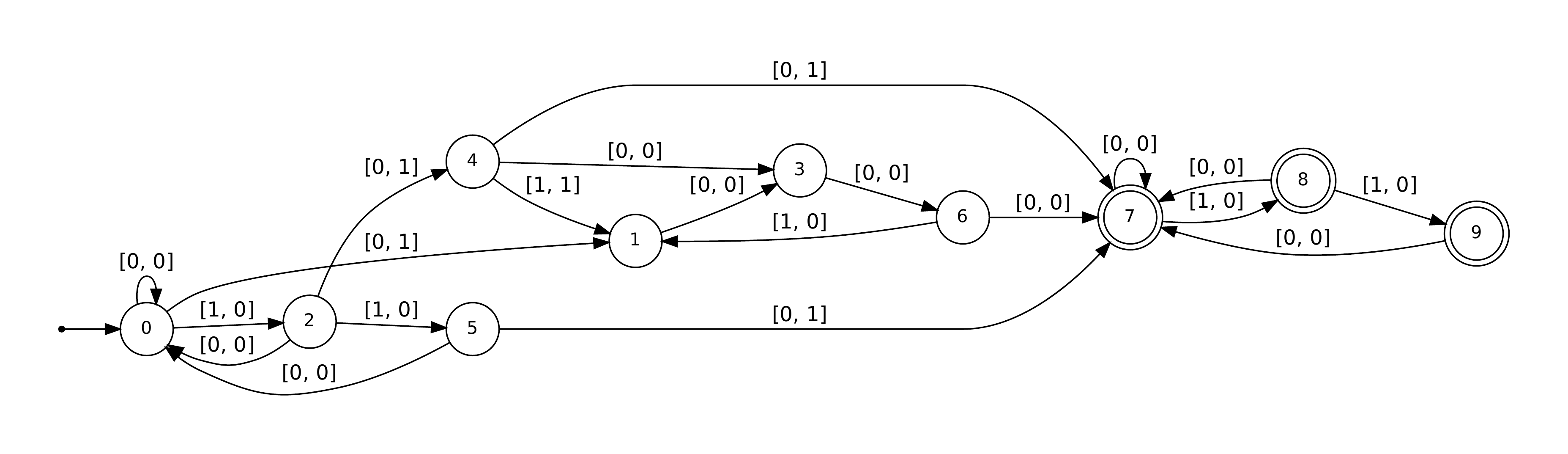}
\caption{Automaton accepting orders and positions of all squares in $\bf T$}
\label{squareorders}
\end{center}
\end{figure}
\end{theorem}

Next, we examine the cubes in $\bf T$.
Evidently Theorem~\ref{squares} implies that any cube in $\bf T$
must be of order $T_n$ or $T_n + T_{n-1}$ for some $n$.
However, not every order occurs.  We thus recover the following
result of Glen \cite[\S 6.3.7]{Glen:2006}.

\begin{theorem}
The cubes in $\bf T$ are of order $T_n$ for $n \geq 5$, and a cube of
each such order occurs.
\end{theorem}

\begin{proof}
We use the predicate
$$ (n>0) \ \wedge \ \exists i \ \forall j \ ((i\leq j) \wedge (j< i+2n))
        \implies {\bf T}[j] = {\bf T}[j+n] .$$
When we run our program, we obtain an automaton accepting exactly the
language $(1000)0^*$, which corresponds to $T_n$ for $n \geq 5$.
\begin{verbatim}
\end{verbatim}
The largest intermediate automaton had 60743 states.
\end{proof}

Next, we encode the orders and positions of all cubes.  We build a 
DFA accepting the language
$$ 
\{ (i,n)_T \ : \ 
(n>0) \ \wedge \ \forall j \ ((i\leq j) \wedge (j< i+2n))
        \implies {\bf T}[j] = {\bf T}[j+n]
\} . $$

\begin{theorem}
The language
$$ \{ (n,i)_T \ : \ \text{there is a cube of order $n$ beginning at
	position $i$ in {\bf T}} \}$$
is accepted by the automaton in Figure~\ref{cubeorders}.

\begin{figure}[H]
\begin{center}
\includegraphics[width=6.5in]{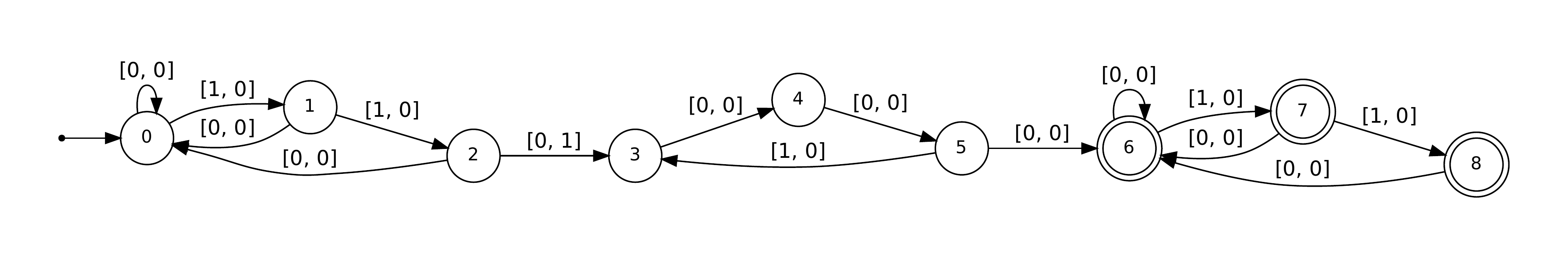}
\caption{Automaton accepting orders and positions of all cubes in $\bf T$}
\label{cubeorders}
\end{center}
\end{figure}
\label{cubes}
\end{theorem}

We also computed an automaton accepting those pairs $(p,n)$ such that
there is a factor of $\bf T$ having length $n$ and period $p$, and
$n$ is the largest such length corresponding to the period $p$.
However, this automaton has 266 states, so we do not give it here.

\subsection{Palindromes}

We now turn to a characterization of the palindromes in $\bf T$.
Once again it turns out that the predicate we previously used
in \cite{Du&Mousavi&Schaeffer&Shallit:2014}, namely,
$$ \exists i \ \forall j<n \ {\bf T}[i+j] = {\bf T}[i+n-1-j], $$
resulted in an intermediate NFA of 5711 states that we could not
successfully determinize.

Instead, we used two equivalent predicates.  The first
accepts $n$ if there is an even-length palindrome,
of length $2n$, centered at position $i$:
$$ \exists i \geq n \ \forall j<n\ {\bf T}[i+j]={\bf T}[i-j-1] .$$
The second accepts $n$ if there is an odd-length palindrome,
of length $2n+1$,
centered at position $i$:
$$ \exists i \geq n \ \forall j\ (1 \leq j \leq n) \implies
{\bf T}[i+j]={\bf T}[i-j].$$

\begin{theorem}
There exist palindromes of every length $\geq 0$ in $\bf T$.
\end{theorem}

\begin{proof}
For the first predicate,
our program outputs the automaton below.  It clearly accepts the
Tribonacci representations for all $n$.
\begin{figure}[H]
\begin{center}
\includegraphics[width=4.8in]{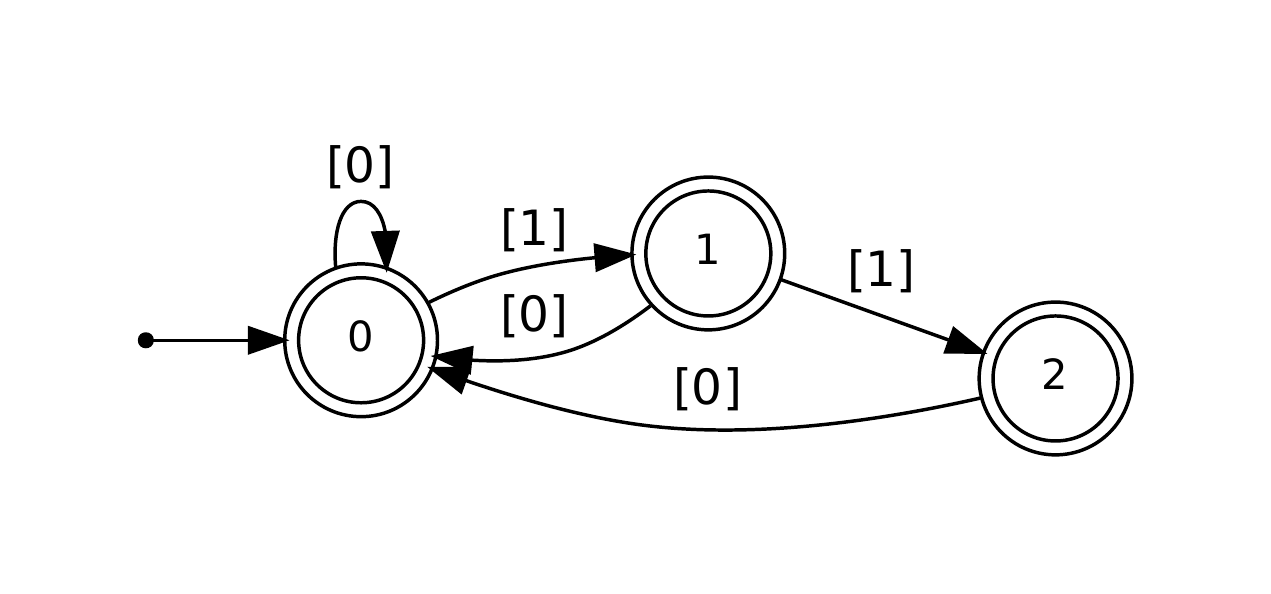}
\caption{Automaton accepting lengths of
palindromes in $\bf T$}
\label{palindrome-orders1}
\end{center}
\end{figure}
The log of our program follows.
\begin{verbatim}
i >= n with 13 states, in 34ms
 j < n with 13 states, in 8ms
  i + j with 150 states, in 53ms
   i - 1 with 7 states, in 155ms
    i - 1 - j with 150 states, in 166ms
     TR[i + j] = TR[i - 1 - j] with 664 states, in 723ms
      j < n => TR[i + j] = TR[i - 1 - j] with 3312 states, in 669ms
       Aj j < n => TR[i + j] = TR[i - 1 - j] with 24 states, in 5782274ms
        i >= n & Aj j < n => TR[i + j] = TR[i - 1 - j] with 24 states, in 0ms
         Ei i >= n & Aj j < n => TR[i + j] = TR[i - 1 - j] with 4 states, in 6ms
overall time: 5784088ms
\end{verbatim}
The largest intermediate automaton had 918871 states.  This was a fairly
significant computation, taking about two hours' CPU time on a laptop.

We omit the details of the computation for the odd-length palindromes,
which are quite similar.
\end{proof}

\begin{remark}
A. Glen has pointed out to us that this follows from the fact that
$\bf T$ is episturmian and hence rich, so a new palindrome is introduced
at each new position in $T$.
\end{remark}

We could also characterize the positions of all 
nonempty palindromes.  To illustrate the idea, we
generated an automaton accepting
$(i,n)$ such that ${\bf T}[i-n..i+n-1]$ is an (even-length) palindrome.

\begin{figure}[H]
\begin{center}
\includegraphics[width=6.5in]{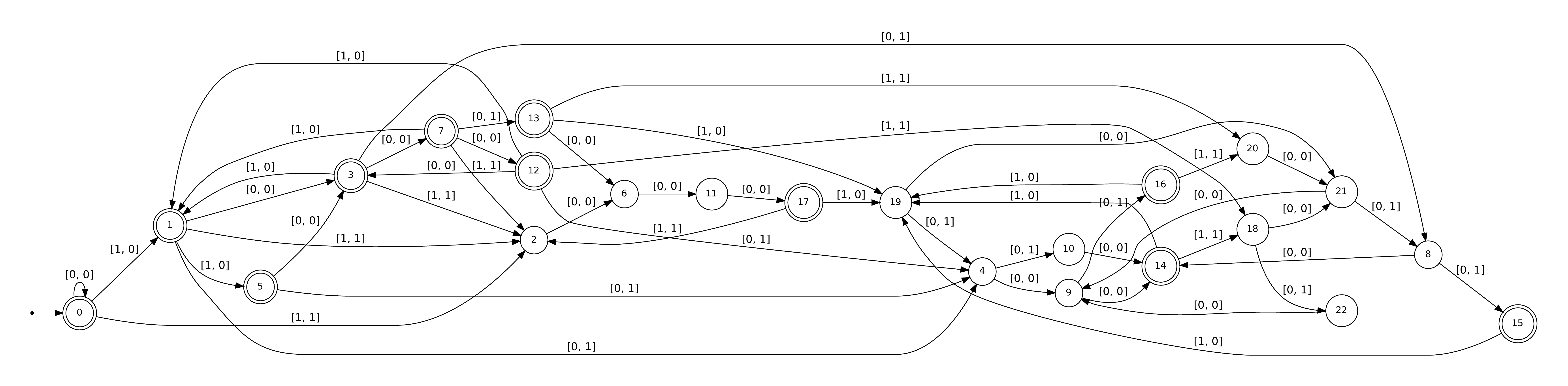}
\caption{Automaton accepting orders and positions of all 
nonempty even-length palindromes in $\bf T$}
\label{palindrome-orders2}
\end{center}
\end{figure}

The prefixes are factors of particular interest.  Let us determine
which prefixes are palindromes:

\begin{theorem}
The prefix ${\bf T}[0..n-1]$ of length $n$ is a palindrome if and
only if $n=0$ or $(n)_T \in 1 + 11 + 10(010)^*(00+001+0011)$.  
\end{theorem}

\begin{proof}
We use the predicate
$$ \forall i<n\ {\bf T}[i] = {\bf T}[n-1-i].$$
The automaton generated is given below.
\begin{figure}[H]
\begin{center}
\includegraphics[width=6.5in]{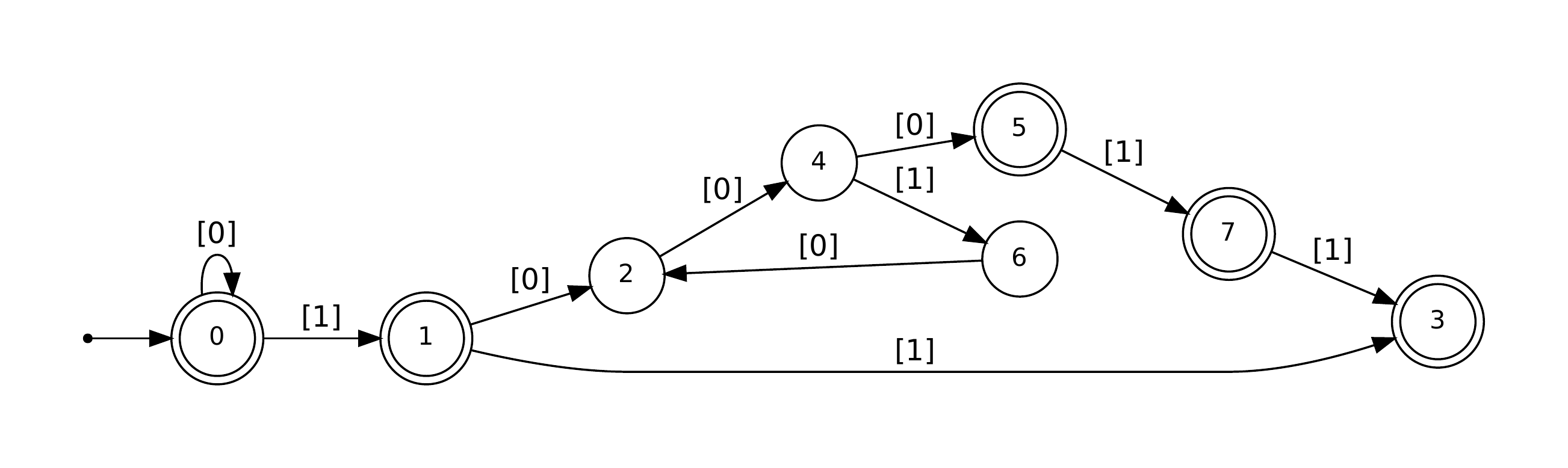}
\caption{Automaton accepting lengths of
palindromes in $\bf T$}
\label{palindrome-orders3}
\end{center}
\end{figure}
\end{proof}

\begin{remark}  A. Glen points out to us that the palindromic
prefixes of $\bf T$ are precisely those of the form $\Pal(w)$, where
$w$ is a finite prefix of the infinite word $(012)^\omega$ and
$\Pal$ denotes the ``iterated palindromic closure''; see, for example,
\cite[Example 2.6]{Glen&Justin:2009}.  She also points out that these
lengths are precisely the integers
$(T_i + T_{i+2} - 3)/2$ for $i \geq 1$.  
\end{remark}

\subsection{Quasiperiods}

We now turn to quasiperiods.  An infinite word $\bf a$ is said to be
{\it quasiperiodic} if there is some finite nonempty word $x$ such that
${\bf a}$ can be completely ``covered'' with translates of $x$.  Here
we study the stronger version of quasiperiodicity where the first copy
of $x$ used must be aligned with the left edge of $\bf w$ and is not
allowed to ``hang over''; these are called {\it aligned covers} in
\cite{Christou&Crochemore&Iliopoulos:2012}.  More precisely, for us
${\bf a} = a_0 a_1 a_2 \cdots$ is quasiperiodic if there exists
$x$ such that for all $i \geq 0$ there exists $j\geq 0$ with $i-n < j \leq i$
such that $a_j a_{j+1} \cdots a_{j+n-1} = x$, where $n = |x|$.
Such an $x$ is called a {\it quasiperiod}.
Note that the condition $j \geq 0$ implies that, in this interpretation,
any quasiperiod must actually be a prefix of $\bf a$.

Glen, Lev\'e, and Richomme characterized the quasiperiods of a large
class of words, including the Tribonacci word \cite[Thm.~4.19]{Glen&Leve&Richomme:2008}.  However, their characterization did not explicitly give
the lengths of the quasiperiods.  We do that in the following result.

\begin{theorem}
A nonempty length-$n$ prefix of $\bf T$ is a quasiperiod of $\bf T$ if
and only if $n$ is accepted by the following automaton:
\begin{figure}[H]
\begin{center}
\includegraphics[width=6.5in]{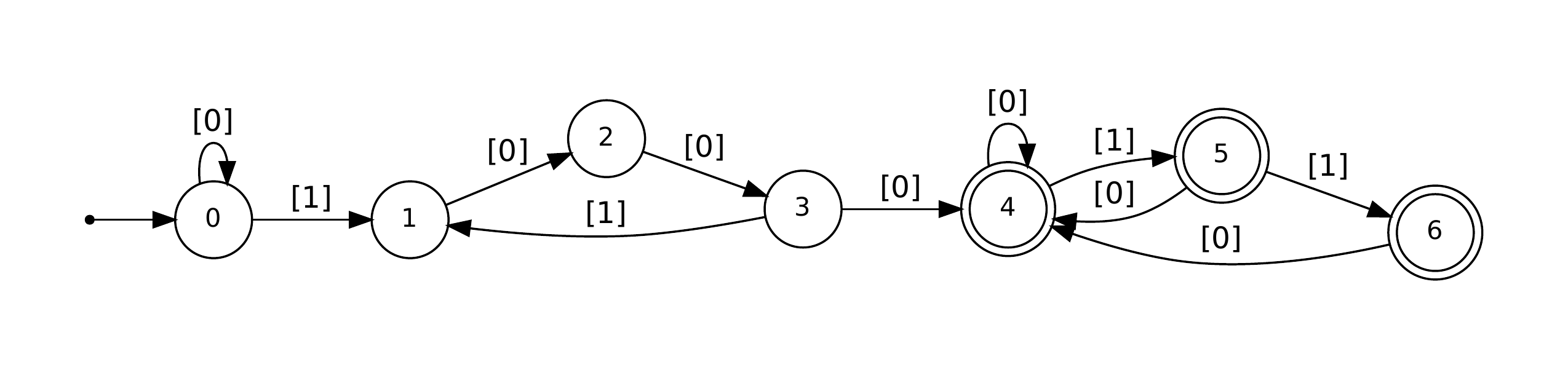}
\caption{Automaton accepting lengths of quasiperiods of the
Tribonacci sequence}
\label{quasitrib}
\end{center}
\end{figure}
\end{theorem}

\begin{proof}
We write a predicate for the assertion that the length-$n$ prefix is 
a quasiperiod:
$$\forall i \geq 0 \ \exists j \text{ with }  i-n < j \leq i
\text{ such that } \forall t<n \ {\bf T}[t] = {\bf T}[j+t] .$$
When we do this, we get the automaton above.
These numbers are those $i$ for which
$T_n \leq i \leq U_n$ for $n \geq 5$,
where $U_2 = 0$, $U_3 = 1$, $U_4 = 3$, and 
$U_n = U_{n-1} + U_{n-2} + U_{n-3} + 3$ for $n \geq 5$. 
\end{proof}

\subsection{Unbordered factors}

Next we look at unbordered factors.  A word $y$ is said to be a {\it
border} of $x$ if $y$ is both a nonempty proper prefix and suffix of
$x$.  A word $x$ is {\it bordered} if it has at least one border.  It
is easy to see that if a word $y$ is bordered iff it has a border of
length $\ell$ with $0 < \ell \leq |y|/2$.

\begin{theorem}
There is an unbordered factor of length $n$ of $\bf T$ if and only if
$(n)_T$ is accepted by the automaton given below.
\begin{figure}[H]
\begin{center}
\includegraphics[width=6.5in]{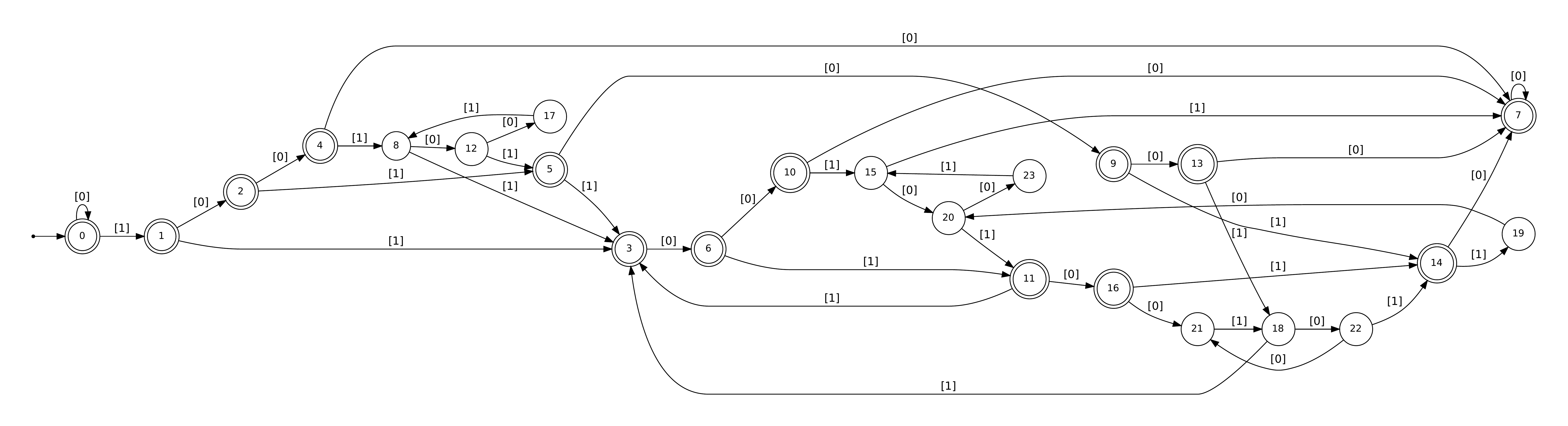}
\caption{Automaton accepting lengths of unbordered factors of the
Tribonacci sequence}
\label{tribunb}
\end{center}
\end{figure}
\end{theorem}

\begin{proof}
As in a previous paper
\cite{Du&Mousavi&Schaeffer&Shallit:2014}
we can express the property of having an unbordered factor of length $n$
as follows
$$ \exists i\ \forall j, 1 \leq j \leq n/2, \
	\exists t<j\ {\bf T}[i+t] \not= {\bf T}[i+n-j+t] .$$
However, this does not run to completion within the available space on
our prover.
Instead,
make the substitutions
$t' = n-j$ and $u = i+t$.   This gives the predicate
$$ \exists i\ \forall t',\ n/2 \leq t' <n, \
        \exists u, \ (i \leq u < i+n-t')\ {\bf T}[u] \not= {\bf T}[u+t'] .$$

Here is the log:
{\footnotesize
\begin{verbatim}
2 * t with 61 states, in 276ms
n <= 2 * t with 79 states, in 216ms
t < n with 13 states, in 3ms
n <= 2 * t & t < n with 83 states, in 9ms
u >= i with 13 states, in 7ms
i + n with 150 states, in 27ms
i + n - t with 1088 states, in 7365ms
u < i + n - t with 1486 states, in 6041ms
u >= i & u < i + n - t with 1540 states, in 275ms
u + t with 150 states, in 5ms
TR[u] != TR[u + t] with 102 states, in 22ms
u >= i & u < i + n - t & TR[u] != TR[u + t] with 7489 states, in 3364ms
Eu u >= i & u < i + n - t & TR[u] != TR[u + t] with 552 states, in 5246873ms
n <= 2 * t & t < n => Eu u >= i & u < i + n - t & TR[u] != TR[u + t] with 944 states, in 38ms
At n <= 2 * t & t < n => Eu u >= i & u < i + n - t & TR[u] != TR[u + t] with 47 states, in 1184ms
Ei At n <= 2 * t & t < n => Eu u >= i & u < i + n - t & TR[u] != TR[u + t] with 25 states, in 2ms
overall time: 5265707ms
\end{verbatim}
}
\end{proof}

\subsection{Lyndon words}

Next, we turn to some results about Lyndon words.  Recall that a
nonempty word $x$ is a {\it Lyndon word\/} if it is lexicographically
less than all of its nonempty proper prefixes.\footnote{There is also a version
where ``prefixes'' is replaced by ``suffixes''.}  

\begin{theorem} There is a factor of length $n$ of $\bf T$ that is
Lyndon if and only if $n$ is accepted by the automaton given below.
\begin{figure}[H]
\begin{center}
\includegraphics[width=6.5in]{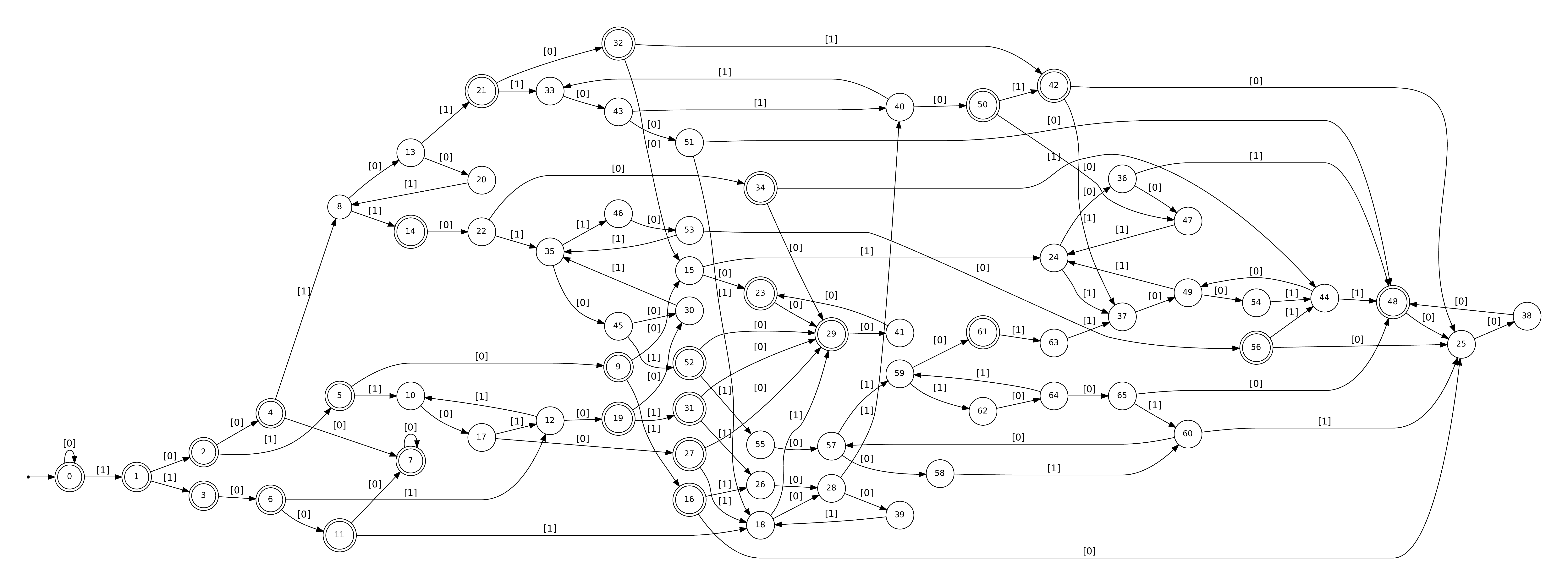}
\caption{Automaton accepting lengths of Lyndon factors of the
Tribonacci sequence}
\label{triblyndon}
\end{center}
\end{figure}
\end{theorem}

\begin{proof}
Here is a predicate 
specifying that there is a factor of length $n$ that is Lyndon:
$$
\exists i\ \forall j, 1 \leq j < n, \ 
\exists t < n-j \ (\forall u<t \ {\bf T}[i+u]={\bf T}[i+j+u]) \ \wedge \ 
{\bf T}[i+t] < {\bf T}[i+j+t] .$$
Unfortunately this predicate did not run to completion, so we
substituted $u' := i+u$ to get
$$
\exists i\ \forall j, 1 \leq j < n, \ 
\exists t < n-j \ (\forall u', i \leq u' < i+t \ {\bf T}[u']={\bf T}[u'+j]) \ \wedge \ 
{\bf T}[i+t] < {\bf T}[i+j+t] .$$
\end{proof}

\subsection{Critical exponent}

Recall from Section~\ref{proofsf} that $\exp(w) = |w|/P$, where $P$ is the
smallest period of $w$.  The {\it critical exponent} of an infinite
word $\bf x$ is the supremum, over all factors $w$ of $\bf x$, of
$\exp(w)$.  

Then Tan and Wen \cite{Tan&Wen:2007} proved that

\begin{theorem}
The critical exponent of $\bf T$ is $\rho \doteq 3.19148788395311874706$,
the real zero of the polynomial $2x^3-12x^2+22x-13$.
\end{theorem}

A. Glen points out that
this result can also be deduced from \cite[Thm.~5.2]{Justin&Pirillo:2002}.

\begin{proof}
Let $x$ be any factor of exponent $\geq 3$ in $\bf T$.  From
Theorem~\ref{cubes} we know that such $x$ exist.  Let $n = |x|$ and
$p$ be the period, so that $n/p \geq 3$.  Then by considering the
first $3p$ symbols of $x$, which form a cube, we have by Theorem~\ref{cubes}
that $p = T_n$.  So it suffices to determine the largest $n$ 
corresponding to every $p$ of the form $T_n$.  We did this using the predicate

		


\begin{figure}[H]
\begin{center}
\includegraphics[width=6.5in]{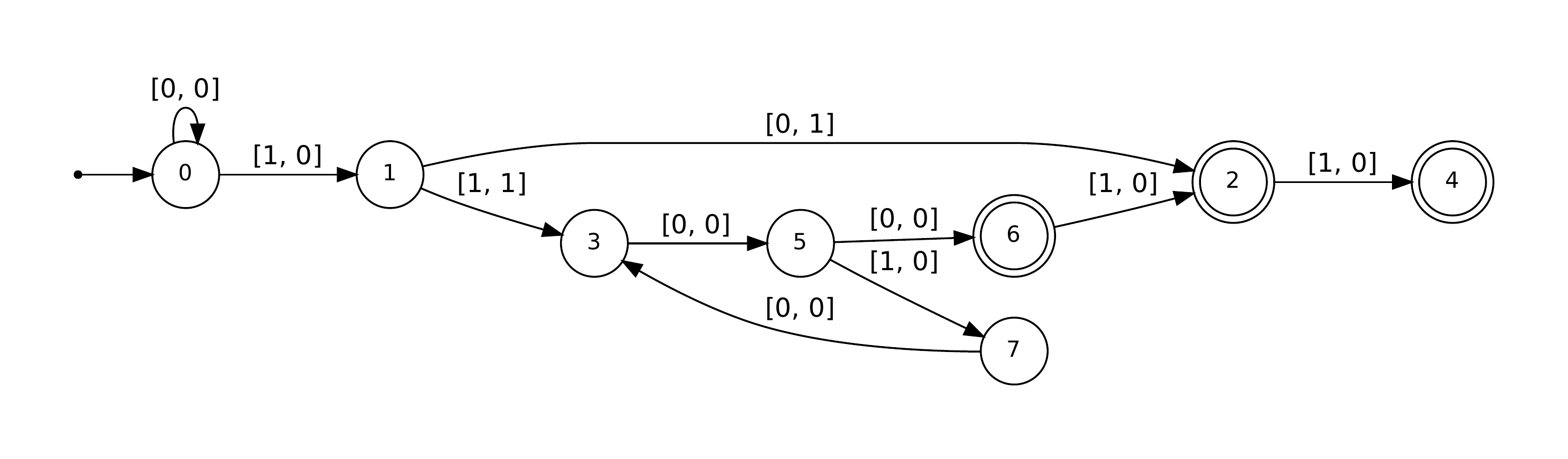}
\caption{Length $n$ of longest factors having period $p= T_n$ of Tribonacci
sequence}
\label{tribmp}
\end{center}
\end{figure}

From inspection of the automaton, we see that the maximum length
of a factor $n = U_j$ having period $p = T_j$, $j \geq 2$, is given by
$$
U_j = \begin{cases}
	2, & \text{if $j = 2$}; \\
	5, & \text{if $j = 3$}; \\
	[110(100)^{i-1} 0]_T, & \text{if $j = 3i+1 \geq 4$}; \\
	[110(100)^{i-1} 01]_T, & \text{if $j = 3i+2 \geq 5$}; \\
	[110(100)^{i-1} 011]_T, & \text{if $j = 3i+3 \geq 6$}.
	\end{cases}
$$
A tedious induction shows that $U_j$ satisfies the linear
recurrence $U_j = U_{j-1}+U_{j-2}+U_{j-3} + 3$ for
$j \geq 5$.  Hence we can write $U_j$ as a linear combination 
Tribonacci sequences and the constant sequence $1$, and solving
for the constants we get
$$ U_j = {5 \over 2}T_j + T_{j-1} + {1 \over 2} T_{j-2} - {3 \over 2}$$
for $j \geq 2$.

The critical exponent of $T$ is then $\sup_{j \geq 1} U_j/T_j$.
Now 
\begin{align*}
U_j/T_j &= {5 \over 2} + {{T_{j-1}} \over {T_j}} + {{T_{j-2}} \over {2T_j}}
- {3 \over {2T_j}} \\
&= {5 \over 2} + \alpha^{-1} + {1 \over 2} \alpha^{-2} + O(1.8^{-j}).
\end{align*}
Hence $U_j/T_j$ tends to $5/2 + \alpha^{-1} + {1\over 2}\alpha^{-2} = \rho$.
\end{proof}

   We can also ask the same sort of questions about the
{\it initial critical exponent} of a word $\bf w$,
which is the supremum over the exponents of all  prefixes of $\bf w$.

\begin{theorem}
The initial critical exponent of $\bf T$ is $\rho - 1$.
\end{theorem}

\begin{proof}
We create an automaton $M_{\rm ice}$ accepting the language
$$L = \{ (n,p)_T \ : \ {\bf T}[0..n-1] \text{ has least period } p \} .$$
It is depicted in Figure~\ref{ice} below.
An analysis similar to that we gave above for the critical exponent gives
the result.

\begin{figure}[H]
\begin{center}
\includegraphics[width=6.5in]{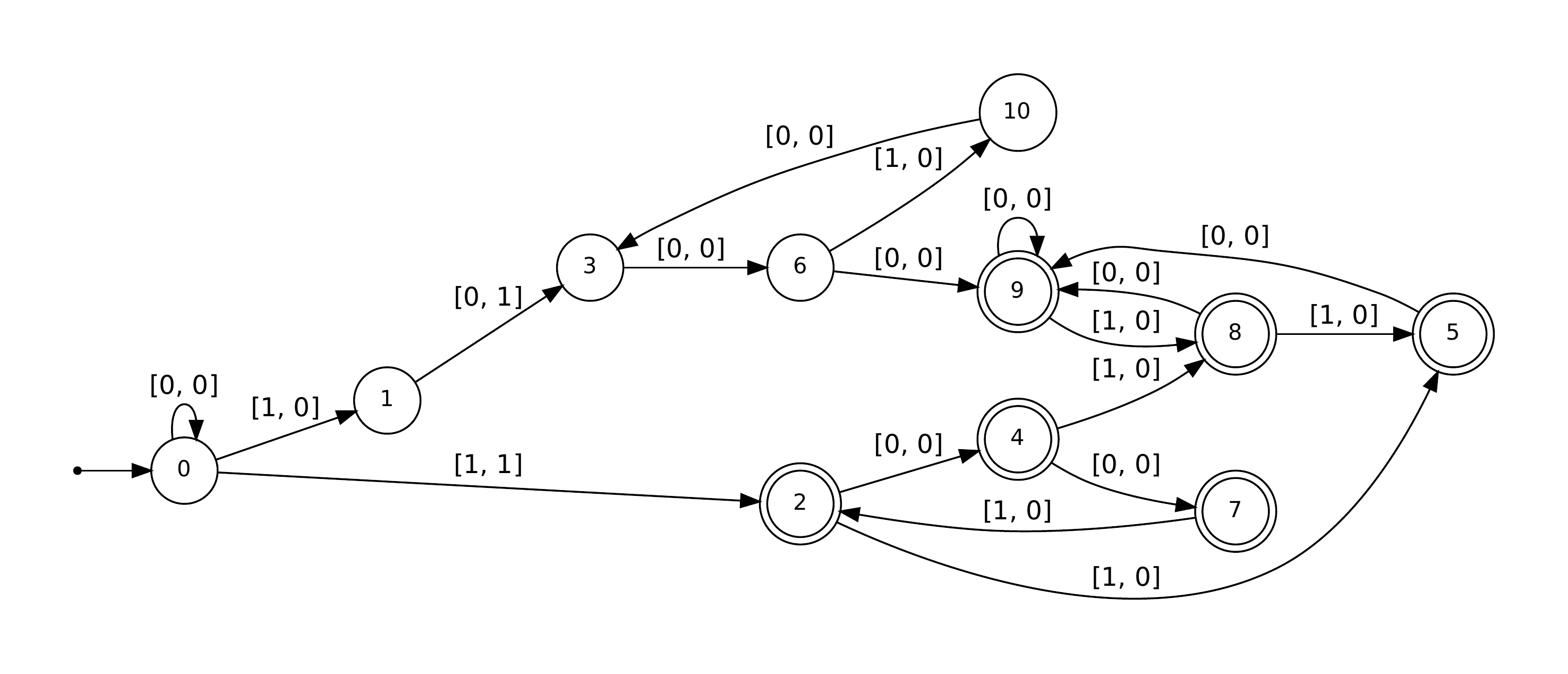}
\caption{Automaton accepting least periods of prefixes of length $n$}
\label{ice}
\end{center}
\end{figure}

\end{proof}

Recall that a primitive word is a non-power; that is, a word that cannot
be written in the form $x^n$ where
$n$ is an integer $\geq 2$.

\begin{theorem}
The only prefixes of the Tribonacci word that are powers are those
of length $2 T_n$ for $n \geq 5$.
\end{theorem}

\begin{proof}
The predicate
$$
\exists d < n\ (\forall j<n-d \ {\bf T}[j]={\bf T}[d+j]) \ \wedge \  
	(\forall k<d \ {\bf T}[k] = {\bf T}[n-d+k])
$$
asserts that the prefix ${\bf T}[0..n-1]$ is a power.
When we run this through our program, the resulting automaton
accepts $100010^*$, which corresponds to $F_{n+1} + F_{n-3} = 2T_n$ for 
$n \geq 5$.  
\end{proof}

\section{Enumeration}
\label{enumer}

Mimicking the base-$k$ ideas in
\cite{Charlier&Rampersad&Shallit:2012}, we can also mechanically enumerate
many aspects of Tribonacci-automatic sequences.  We do this
by encoding the factors having the
property in terms of paths of an automaton.  This gives the
concept of {\it Tribonacci-regular sequence} 
Roughly speaking, a sequence
$(a(n))_{n \geq 0}$ taking values in $\Enn$ is Tribonacci-regular
if the set of sequences
$$ \{ (a([xw]_T)_{w \in \Sigma_2^*} \ : \ x \in \Sigma_2^* \} $$
is finitely generated.  Here we assume that $a([xw]_T)$ evaluates
to $0$ if $xw$ contains the word $111$.  Every Tribonacci-regular
sequence $(a(n))_{n \geq 0}$ has a {\it linear representation}
of the form $(u, \mu, v)$ where $u$ and $v$ are row and column
vectors, respectively, and $\mu:\Sigma_2 \rightarrow \Enn^{d \times d}$
is a matrix-valued morphism, where $\mu(0) = M_0$ and $\mu(1) = M_1$
are $d \times d$ matrices for some $d \geq 1$, such that
$$a(n) = u \cdot \mu(x) \cdot v$$
whenever $[x]_T = n$.  The {\it rank} of the representation is
the integer $d$.

Recall that if $\bf x$ is an infinite word, then the subword
complexity function $\rho_{\bf x} (n)$ counts the number of
distinct factors of length $n$.  Then, in analogy with
\cite[Thm.~27]{Charlier&Rampersad&Shallit:2012}, we have

\begin{theorem}
If $\bf x$ is Tribonacci-automatic, then the subword complexity
function of $\bf x$ is Tribonacci-regular.
\end{theorem}


Using our implementation, we can obtain a linear representation
of the subword complexity function for $\bf T$.  An obvious choice
is to use the language
$$ \{ (n,i)_T \ : \ \forall j < i \ {\bf T}[i..i+n-1] \not=
{\bf T}[j..j+n-1] \} ,$$  
based on a predicate
that expresses the assertion that the factor of length
$n$ beginning at position $i$ has never appeared before.
Then, for each $n$, the number of corresponding $i$ gives
$\rho_{\bf T}(n)$.

However, this does not run to completion in our implementation in
the allotted time and space.  Instead, let us
substitute $u = j+t$ and and $k = i-j$ to get the predicate
$$ \forall k\ (((k>0)\wedge (k\leq i)) \implies
(\exists u\ ((u\geq j) \wedge (u<n+j)\wedge ({\bf T}[u]\not= {\bf T}[u+k])))) .
$$
This predicate is close to the upper limit of what we can compute
using our program.  The largest intermediate automaton had
1230379 states and the program took 12323.82 seconds, giving us
a 
linear representation $(u, \mu, v)$ 
rank $22$.  When we minimize this using the algorithm
in \cite{Berstel&Reutenauer:2011} we get the rank-$12$ linear
representation
\begin{align*}
u &= [1\ 0 \ 0 \ 0 \ 0 \ 0 \ 0 \ 0 \ 0 \ 0 \ 0 \ 0 ] \\
M_0 &= \left[ \begin{array}{cccccccccccc}
1&  0&  0&  0&  0&  0&  0&  0&  0&  0&  0&  0\\
0&  0&  1&  0&  0&  0&  0&  0&  0&  0&  0&  0\\
0&  0&  0&  0&  1&  0&  0&  0&  0&  0&  0&  0\\
-1&  0&  1&  0&  1&  0&  0&  0&  0&  0&  0&  0\\
0&  0&  0&  0&  0&  0&  1&  0&  0&  0&  0&  0\\
-1&  0&  1&  0&  0&  0&  1&  0&  0&  0&  0&  0\\
-2&  0&  1&  0&  1&  0&  1&  0&  0&  0&  0&  0\\
-3&  0&  2&  0&  1&  0&  1&  0&  0&  0&  0&  0\\
-4&  0&  2&  0&  2&  0&  1&  0&  0&  0&  0&  0\\
-5&  0&  2&  0&  2&  0&  2&  0&  0&  0&  0&  0\\
-6&  0&  2&  0&  3&  0&  2&  0&  0&  0&  0&  0\\
-10&  0&  3&  0&  4&  0&  4&  0&  0&  0&  0&  0\\
\end{array}
\right] \\
M_1 &= \left[ \begin{array}{cccccccccccc}
0&1&0&0&0&0&0&0&0&0&0&0\\
0&0&0&1&0&0&0&0&0&0&0&0\\
0&0&0&0&0&1&0&0&0&0&0&0\\
0&0&0&0&0&0&0&0&0&0&0&0\\
0&0&0&0&0&0&0&1&0&0&0&0\\
0&0&0&0&0&0&0&0&1&0&0&0\\
0&0&0&0&0&0&0&0&0&1&0&0\\
0&0&0&0&0&0&0&0&0&0&1&0\\
0&0&0&0&0&0&0&0&0&0&0&0\\
0&0&0&0&0&0&0&0&0&0&0&1\\
0&0&0&0&0&0&0&0&0&0&0&0\\
0&0&0&0&0&0&0&0&0&0&0&0\\
\end{array}
\right] \\
v' &= [1\ 3\ 5\ 7\ 9\ 11\ 15\ 17\ 21\ 29\ 33\ 55]^R \\
\end{align*}

Comparing this to an independently-derived linear representation of
the function $2n+1$, we see they are the same.  From
this we get a well-known result 
(see, e.g., \cite[Thm.~7]{Droubay&Justin&Pirillo:2001}):

\begin{theorem}
The subword complexity function of $\bf T$ is $2n+1$.
\label{sturmcomp}
\end{theorem}

We now turn to computing the exact number of square occurrences 
in the finite Tribonacci words $Y_n$.

To solve this using our approach, we first generalize the problem to
consider any length-$n$ prefix of $Y_n$, and not simply the prefixes
of length $T_n$.

The following predicate represents
the number of distinct squares in ${\bf T}[0..n-1]$:
\begin{multline*}
L_{\rm ds} :=
\{ (n,i,j)_T \ : \ (j \geq 1) \text{ and } (i+2j \leq n) \text{ and }
	{\bf T}[i..i+j-1] = {\bf T}[i+j..i+2j-1]  \\
\text{ and } \forall i' < i \ 
{\bf T}[i'..i'+2j-1] \not= {\bf T}[i..i+2j-1] \} .
\end{multline*}
This predicate asserts that ${\bf T}[i..i+2j-1]$ is a square 
occurring in ${\bf T}[0..n-1]$ and 
that furthermore it is the first occurrence of this particular
word in ${\bf T}[0..n-1]$.

The second represents the total number of occurrences of squares
in ${\bf T}[0..n-1]$:
$$ L_{\rm dos} := \{ (n,i,j)_T \ : \ (j \geq 1) \text{ and }
(i+2j \leq n) \text{ and }
        {\bf T}[i..i+j-1] = {\bf T}[i+j..i+2j-1] \} .$$
This predicate asserts that ${\bf T}[i..i+2j-1]$ is a square 
occurring in ${\bf T}[0..n-1]$.

Unfortunately, applying our enumeration method to this
suffers from the same problem as before, so we rewrite
it as
$$ (j \geq 1) \wedge \ (i+2j \leq n) \ \wedge 
	\forall u \ ((u \geq i)\wedge(u<i+j)) \implies {\bf T}[u] = {\bf T}[u+j] $$
When we compute the linear representation of the function counting the
number of such $i$ and $j$, we get a linear representation of
rank $63$.  Now we compute the minimal polynomial of $M_0$
which is $(x-1)^2 (x^2 + x + 1)^2 (x^3 - x^2 -x - 1)^2$.  
Solving a linear system in terms of the roots
(or, more accurately, in terms of the sequences $1$, $n$,
$T_n$, $T_{n-1}$, $T_{n-2}$, $nT_n$, $nT_{n-1}$, $n T_{n-2}$) gives

\begin{theorem}
The total number of occurrences of squares in the Tribonacci word
$Y_n$ is 
$$c(n) = {n \over {22}}(9 T_n - T_{n-1} - 5 T_{n-2}) + 
	{1 \over {44}} (-117 T_n + 30 T_{n-1} + 33 T_{n-2}) + n - {7 \over 4} 
	$$
for $n \geq 5$.
\end{theorem}
	 
In a similar way, we can count the occurrences of cubes in the 
finite Tribonacci word $Y_n$.
Here we get a linear representation of rank 46.  The minimal
polynomial for $M_0$ is 
$x^4(x^3-x^2-x-1)^2(x^2+x+1)^2(x-1)^2$.
Using analysis exactly like the square case, we easily find

\begin{theorem}
Let $C(n)$ denote the number of cube occurrences in the Tribonacci
word $Y_n$.  Then for $n \geq 3$ we have
\begin{multline*}
C(n) = {1 \over {44}} (T_n + 2T_{n-1} -33 T_{n-2}) +
{n \over {22}} (-6T_n + 8T_{n-1} + 7T_{n-2}) + {n \over 6} \\
-{1 \over 4} [n \equiv \modd{0} {3}] + {1 \over 12} [n \equiv \modd{1} {3}]
-{7 \over 12} [n \equiv \modd{2} {3}] .
\end{multline*}
Here $[P]$ is Iverson notation, and equals $1$ if $P$ holds and $0$
otherwise.
\end{theorem}

\section{Other words}

Of course, our technique can also prove things about words other than
$\bf T$.  For example, consider the binary Tribonacci word
${\bf b} = 0101010010101010101001010101 \cdots $ obtained from $\bf T$ by mapping each letter $i$ to
$\min(i,1)$.  

\begin{theorem}
The critical exponent of $\bf b$ is $13/2$.
\end{theorem}

\begin{proof}
We use our method to verify that $\bf b$ has $(13/2)$-powers and no larger
ones.  (These powers arise only from words of period $2$.)
\end{proof}

\section{Abelian properties}

We can derive some results about the abelian properties of the
Tribonacci word $\bf T$ by proving the analogue of 
Theorem 63 of \cite{Du&Mousavi&Schaeffer&Shallit:2014}:

\begin{theorem}
Let $n$ be a non-negative integer and let
$e_1 e_2 \cdots e_j$ be a Tribonacci representation of $n$,
possibly with leading zeros, with $j \geq 3$.  Then
\begin{itemize}
\item[(a)] $|{\bf T}[0..n-1]|_0 = [e_1 e_2 \cdots e_{j-1}]_T + e_j$ .
\item[(b)] $|{\bf T}[0..n-1]|_1 = [e_1 e_2 \cdots e_{j-2}]_T + e_{j-1}$ .
\item[(c)] $|{\bf T}[0..n-1]|_2 = [e_1 e_2 \cdots e_{j-3}]_T + e_{j-2}$.
\end{itemize}
\label{tribab}
\end{theorem}

\begin{proof}
By induction, in analogy with the proof of
\cite[Theorem 63]{Du&Mousavi&Schaeffer&Shallit:2014}.
\end{proof}

Recall that the Parikh vector $\psi(x)$ of a word $x$ over an
ordered alphabet $\Sigma = \lbrace a_1, a_2, \ldots, a_k \rbrace$
is defined to be $(|x|_{a_1}, \ldots, |x|_{a_k})$,
the number of occurrences of each letter in $x$.
Recall that the abelian complexity function $\rho_{\bf w}^{\rm ab} (n)$ counts
the number of distinct Parikh vectors of the length-$n$ factors of an
infinite word $\bf w$.

Using Theorem~\ref{tribab} we get another proof of a recent
result of Turek \cite{Turek:2013}.

\begin{corollary}
The abelian complexity function of $\bf T$ is Tribonacci-regular.
\end{corollary}

\begin{proof}
First, from Theorem~\ref{tribab} there exists an automaton $\tabb$
such that $(n,i,j,k)_T$ is accepted iff $n = \psi({\bf T}[0..n-1])$.
In fact, such an automaton has 32 states.

Using this automaton, we can
create a predicate $P(n,i)$ such that the
number of $i$ for which $P(n,i)$ is true equals $\rho_{\bf T}^{\rm ab}
(n)$.  For this we assert that $i$ is the least index at which
we find an occurrence of the Parikh vector of ${\bf T}[i..i+n-1]$:
\begin{multline*}
\forall i' < i \ \exists a_0, a_1, a_2, b_0,b_1, b_2, 
c_0,c_1,c_2, d_0, d_1, d_2 \\
\tabb(i+n,a_0,a_1,a_2) \ \wedge \ \tabb(i,b_0,b_1,b_2) \ \wedge \ 
\tabb(i'+n,c_0,c_1,c_2) \ \wedge \ \tabb(i',d_0,d_1,d_2) \ \wedge \  \\
((a_0 - b_0 \not= c_0 - d_0) \ \vee \ 
(a_1 - b_1 \not= c_1 - d_1) \ \vee \ 
(a_2 - b_2 \not= c_2 - d_2)) .
\end{multline*}
\end{proof}

\begin{remark}
Note that exactly the same proof would work for any word and
numeration system
where the Parikh vector of prefixes of length $n$ is ``synchronized'' with
$n$.
\end{remark}

\begin{remark}
In principle we could mechanically compute the Tribonacci-regular representation
of the abelian complexity function using this technique,
but with our current implementation
this is not computationally feasible.
\end{remark}

\begin{theorem}
Any morphic image of the Tribonacci word is Tribonacci-automatic.
\end{theorem}

\begin{proof}
In analogy with Corollary 69 of \cite{Du&Mousavi&Schaeffer&Shallit:2014}.
\end{proof}

\section{Things we could not do yet}

There are a number of things we have not succeeded in computing with
our prover because it ran out of space.  These include

\begin{itemize}
\item mirror invariance of $\bf T$ (that is, if $x$ is a finite factor
then so is $x^R$);

\item Counting the number of special factors of length $n$ (although it can be deduced
from the subword complexity function);

\item statistics about, e.g, lengths of squares, cubes, etc.,
in the ``flipped'' Tribonacci sequence
\cite{Rosema&Tijdeman:2005}, the fixed point of
$ 0 \rightarrow 01$, $1 \rightarrow 20$, $2 \rightarrow 0$;

\item recurrence properties of the Tribonacci word;

\item counting the number of distinct squares (not occurrences) in
the finite Tribonacci word $Y_n$.

\item abelian complexity of the Tribonacci word.
\end{itemize}

In the future, an improved implementation may succeed in resolving
these in a mechanical fashion.

\section{Details about our implementation}

Our program is written in JAVA, and was developed using the
{\tt Eclipse} development environment.\footnote{Available from {\tt http://www.eclipse.org/ide/} .} 
We used the {\tt dk.brics.automaton}
package, developed by Anders M{\o}ller at Aarhus University, for
automaton minimization.\footnote{Available from {\tt http://www.brics.dk/automaton/} .}
{\tt Maple 15} was used
to compute characteristic polynomials.\footnote{Available from {\tt http://www.maplesoft.com} .}
The {\tt GraphViz} package was used to display automata.\footnote{Available
from {\tt http://www.graphviz.org} .}  We used a program written in
APL X\footnote{Available from {\tt http://www.microapl.co.uk/apl/ } .}
to implement minimization of linear representations. 

Our program consists of about 2000 lines of code.  We used
Hopcroft's algorithm for DFA minimization.

A user interface is provided to enter queries in a language very
similar to the language of first-order logic. The intermediate and
final result of a query are all automata. At every intermediate step,
we chose to do minimization and determinization, if necessary. Each
automaton accepts tuples of integers in the numeration system of choice.
The built-in numeration systems are ordinary base-$k$ representations,
Fibonacci base, and Tribonacci base.
However, the program can be used with any numeration
system for which an automaton for addition and ordering can be
provided. These numeration system-specific automata can be declared in
text files following a simple syntax. For the automaton resulting from a 
query it is always guaranteed that if a tuple $t$ of integers is
accepted, all tuples obtained from $t$ by addition or truncation of
leading zeros are also accepted. In Tribonacci representation, we make sure that
the accepting integers do not contain three consecutive $1$'s.

The source code and manual will soon be available for free download.


\section{Acknowledgments}

We are very grateful to Amy Glen for her recommendations and advice.

\newcommand{\noopsort}[1]{} \newcommand{\singleletter}[1]{#1}

\end{document}